\documentclass[11pt,letterpaper]{article}
\usepackage{fullpage}
\usepackage[round]{natbib}
\usepackage{amsmath,amsfonts,amssymb,amsthm,  mathrsfs,mathtools}
\usepackage[margin=1in]{geometry}
\usepackage[hyperindex,breaklinks]{hyperref}
\hypersetup{
  colorlinks,
  citecolor=violet,
  linkcolor=blue,
  urlcolor=blue}
\usepackage{dsfont}
\usepackage{breakcites}
\numberwithin{equation}{section}   
\usepackage{booktabs} 
\usepackage[ruled,noend]{algorithm2e} 

\SetAlFnt{\small}
\SetAlCapFnt{\small}
\SetAlCapNameFnt{\small}
\SetAlCapHSkip{0pt}
\IncMargin{-\parindent}

\numberwithin{equation}{section}
\usepackage{macros}
\usepackage{slivkins-theorems}
\usepackage{diagbox}

\title{Regret Analysis of Repeated Delegated Choice\footnote{Alphabetical author ordering. Published in AAAI 2024.}}

  \newcommand{\country}[1]{#1.}
  \newcommand{\city}[1]{#1}
  \newcommand{\institution}[1]{#1}
  \newcommand{\email}[1]{Email: \texttt{#1}}
  \newcommand{\affiliation}{\thanks}
\author{
  {MohammadTaghi Hajiaghayi
 \affiliation{
   \institution{University of Maryland}
   \city{College Park, MD}
   \country{USA}
 \email{\{hajiagha,mahdavi,krezaei,suhoshin\}@umd.edu}
 }}
 \and
 Mohammad Mahdavi\footnotemark[1]
 \and
 Keivan Rezaei\footnotemark[1]
 \and
 Suho Shin\footnotemark[1]
}

\begin{document}

\maketitle

\begin{abstract}
    We present a study on a repeated delegated choice problem, which is the first to consider an online learning variant of Kleinberg and Kleinberg, EC'18. 
    In this model, a principal interacts repeatedly with an agent who possesses an exogenous set of solutions to search for efficient ones. 
    Each solution can yield varying utility for both the principal and the agent, and the agent may propose a solution to maximize its own utility in a selfish manner.
    To mitigate this behavior, the principal announces an eligible set which screens out a certain set of solutions.
    The principal, however, does not have any information on the distribution of solutions in advance.
    Therefore, the principal dynamically announces various eligible sets to efficiently learn the distribution.
    The principal's objective is to minimize cumulative regret compared to the optimal eligible set in hindsight.
    We explore two dimensions of the problem setup, whether the agent behaves myopically or strategizes across the rounds, and whether the solutions yield deterministic or stochastic utility. 
    Our analysis mainly characterizes some regimes under which the principal can recover the sublinear regret, thereby shedding light on the rise and fall of the repeated delegation procedure in various regimes.
\end{abstract}

\section{Introduction}
Delegation is perhaps one of the most frequent economic interactions one may see around in real life \citep{holmstrom1980theory, bendor2001theories, amador2013theory}.
Abstractly speaking, consider a principal with less information who tries to find an optimal solution from an agent with expertise, but there's an information asymmetry such that she\footnote{Feminine pronouns (masculine) hereafter denote the principal (agent).} \emph{cannot directly access} the solutions that the agent possesses \citep{alonso2008optimal, kleinberg2018delegated,kleiner2022optimal, hajiaghayi2023multi}.
Instead, she requires the agent to propose a set of solutions and then commits to the final one among them.
The principal and the agent, however, may have \emph{misaligned utility} for the solution selected, and thus the agent may propose a solution in a selfish manner.
To cope with it, the principal announces a set of eligible solutions before the agent proposes, and only accepts the eligible solution.


To provide a concrete example, consider (online) labor market or crowdsourcing platform such as Upwork.
We have a task requester (principal) who regularly visits the platform (agent) and tries to solve a series of tasks.
The platform has a pool of workers (solutions).
At each time the requester visits, the platform recommends some set of workers, and the requester selects a single worker to commit to the task.
Obviously, the task requester wants to hire a qualified worker.
The platform, on the other hand, aims to maximize its long-term revenue by recommending workers who solve tasks quickly, even if their quality is not high, allowing them to be assigned to other tasks promptly.
This misalignment of utility may lead to the platform strategically recommending unqualified workers. To mitigate this, the requester sets restrictions, such as requiring certificates in specific areas like a foreign language or web development, when requesting worker recommendations. We refer to Appendix A in the full paper for more examples on motivations.

If the task requester is fully aware of the set of workers that the platform has, then she can directly impose a strong restriction to make the platform recommend the specific workers she wants.
In practice, however, such information is not feasible priorly, instead, the requester needs to \emph{learn the distribution} of existing workers in the repeated interaction.
The fundamental question here is, how the requester should dynamically determine which sort of restriction to impart at each round, in order to maximize cumulative utility over the set of tasks.
Furthermore, one may ask what happens if the platform also tries to \emph{strategize across the rounds} to deceive the requester, and what if the quality of each worker is not fixed in advance, but rather is given from a \emph{latent distribution}. 
This work introduces the \emph{repeated delegated choice} problem, which focuses on how the principal can design an efficient delegation mechanism. To the best of our knowledge, this is the first study to explore an online learning extension of the delegated choice problem presented by \cite{armstrong2010model,kleinberg2018delegated,hajiaghayi2023multi}. In our model, the principal lacks initial information about the solutions' distribution. Instead, through repeated announcements of eligible sets that may screen out some solutions, the principal aims to learn the solutions' distributions in a sample-efficient manner. The principal aims to minimize cumulative regret compared to the optimal eligible set in hindsight.


\begin{table}
\centering
\begin{tabular}{c|c|c}
\diagbox{Utility}{Behavior} & Myopic                           & Strategic                                                                                                                     \\ \hline 
Deterministic      & Theorem~\ref{thm:deter_myop} &  
Theorem \ref{thm:alg2_regret},\ref{thm:lips-lns-reg}\\
\hline
Stochastic         & Theorem~\ref{thm:stoc_myop} & Theorem~\ref{thm:stoc_stra} \\ 

\hline
\end{tabular}
\caption{Summary of our results under different settings.}
\label{tab:summary}
\end{table}



We distill the problem into two dimensions of whether the utility of each solution is deterministic or stochastic, and whether the agent strategizes across the rounds or not, and provide a comprehensive regret analysis for each setting.
In the myopic agent setting, the agent plays a best-response to the eligible set at each round, \ie a strategy which maximizes \emph{myopic utility} without regard to the future utility and resulting behavior of the mechanism.\footnote{This model of myopic agent accommodates a perspective of "multiple agents" setting in which at each round an agent having the same type of solutions arrive and interacts with the principal. In this viewpoint, the agents are bound to be myopic due to the single round interaction per agent.}
Hence, the principal's objective boils down to efficiently learning the distribution of utilities by selecting proper eligible sets at each round, while only observing the partial feedback from the choice of eligible set, \ie  which solution the agent submits (or possibly declines to submit any).
This challenge intensifies with a strategic agent, as the agent may intentionally hide solutions or deviate from their best response for greater utility in later rounds. Consequently, the feedback is not guaranteed to be stochastic across rounds, making the analysis more complex.

\paragraph{Our contributions}
We here provide a summary of our contributions and techniques.
All the proofs can be found in the appendix in the full paper.
The results are summarized in Table~\ref{tab:summary}.
First, we observe a revelation-principle-style-of-result such that it suffices to focus on a class of so-called single-proposal mechanism, formally defined in Definition~\ref{def:single}.
Interestingly, we show that the myopic deterministic setting can be reduced to the repeated posted price mechanism (RPPM) with myopic buyer.
\footnote{Overall, we observe an intimate connection between the RPPM and our problem under certain settings. In general, however, our problem spawns additional challenges of having multiple latent random variables and the principal is even unaware of the number of potential solutions. 
We provide more detailed discussion in Appendix B in the full paper.
}
Denoting the principal's utility for each solution by $X_i$, one can effectively construct an instance of RPPM by converting $X_i$ to the buyer's value $v = \max_i X_i$ in RPPM. In both problems, the optimal benchmark is to obtain $\max_i X_i$, and the reduction follows.
Combined further with an iterative algorithm, we obtain a regret upper bound of $O(\min(K, \log \log T))$, where $K$ denotes the number of solutions and $T$ is the time horizon.

With stochastic valuation, however, this does not work since the benchmark in RPPM is to put an ex-ante best fixed price, which does not coincide $\max_i X_i$.
Indeed, we observe that the optimal benchmarks cannot be reduced from one to another in general.
Instead, we mainly reduce our problem to a stochastic multi-armed bandit problem via proper discretization over the space of eligible sets equipped with a variant of analysis by~\cite{kleinberg2003posted}, and obtain a regret of $O(\sqrt{T \log T})$ under the same assumption imposed in~\cite{kleinberg2003posted}.


For the strategic agent with deterministic utility, we first observe that it is necessary to impose a certain assumption on the agent's utility sequence to obtain positive results.
Precisely, the agent with non-discounting utility can strategize so that no algorithm can obtain sublinear regret, where the formal proof is presented in Appendix K.
In this context, to capture both of the practicality and theoretical tractability, we consider $\gamma$-discounting strategic agent whose utility is discounted by a multiplicative factor of $\gamma$ at each round.
We also note that this is common in the literature~\cite{amin2013learning,haghtalab2022learning}.

Given that, for the $\gamma$-discounting agent with deterministic utility, we first consider a case in which the agent's utility is uniformly bounded below by $\ymin$ and the principal is aware of it.
In this setting, by exploiting the delay technique of~\cite{haghtalab2022learning}, we obtain a regret bound of $O(K\tgamma \log \frac{\tgamma}{\ymin})$ where $\tgamma = 1/(1-\gamma)$.
The dependence on $K$ can be replaced by $\log T$ by shrinking the eligible set more in an aggressive manner, thereby obtaining a regret of $O(\tgamma \log \frac{\tgamma}{\ymin} + \log T)$.
Note that these bounds yield sublinear regret only if $\ymin = e^{-o(T)}$.
We complement these results by showing that any algorithm suffers regret of $\Omega(T)$ if $\ymin \le e^{-T}$.

On the other hand, if the agent's minimum utility is not known or unbounded, there's no guarantee that the agent behaves myopically for any delay that is imposed in the algorithm.
Instead, under minor assumptions that the solutions are densely spread with respect to parameter $d$ and Lipschitzness between the principal's and agent's utilities, we obtain an efficient algorithm that achieves a regret upper bound of $O(\tgamma \log \frac{\tgamma}{\alpha} + \log \frac{1}{d} + dT)$, where $\alpha$ is a function of the Lipschitz parameters.
The linear dependence of $O(dT)$ regret may look a lot at first glance, we observe this is inevitable for any algorithm, thereby justifying our assumption. 

In the stochastic setting with $\gamma$-discounting strategic agent, we reuse the machinery by \cite{haghtalab2022learning, lancewicki2021stochastic}, and obtain a regret of $O(\sqrt{T \log T})$.
More specifically, we can view the proposed solution as a perturbed output of a stochastic bandit, where the perturbation comes from the agent’s strategic behavior. 
The technical subtlety lies on how we should upper/lower bound such perturbed output to properly apply \cite{lancewicki2021stochastic}, i.e., how we should construct a random perturbation interval.

\subsection{Related Works}
\paragraph{Delegation}
Dating back to the seminal work of  \cite{holmstrom1980theory}, a number of literature from the economics community study the theory of delegation, mostly within the extent of characterizing the regimes under which some simple mechanisms reach the optimal solution \cite{alonso2008optimal, armstrong2010model, kleiner2022optimal}.
Recently, \cite{kleinberg2018delegated} study a problem of \emph{delegated choice}\footnote{They consider two types of problem settings, one of which is delegated search with sampling costs, and the other is delegated choice, referring back to \cite{armstrong2010model}.
Since we also assume that the solutions of the agent are exogenous to mechanisms, we frame our model as a delegated choice problem.} with a lens of computer science, and show that there exists a mechanism with $2$-approximation compared to the case in which the principal can fully access all the solutions in advance, based on a novel connection to prophet inequality problem~\citep{samuel1984comparison}.
Their result, however, depends on the assumption that the principal knows the distribution from which the utility of each solution is drawn, \ie they study the efficiency of Bayesian mechanism.
Interestingly, if the principal has no such information at all, \ie prior-independent mechanism, the result becomes largely pessimistic, \ie there exists a problem instance in which the principal's approximation becomes arbitrarily bad.
\cite{hajiaghayi2023multi} reveal that prior-independent mechanisms can be made efficient with multiple agents, but this does not hold with a single agent.

\paragraph{Repeated delegation}
\cite{lipnowski2020repeated} study a problem of infinitely repeated delegation, however, their model of delegated choice is largely different from ours.
Mainly, their model considers aligned utility but when the principal bears the cost of adopting a project.
Their objective is to persuade the agent to adopt the project when it is truly good, whereas the agent tries to always adopt the project.
Several lines of work \cite{li2017power, guo2021dynamic} study a repeated game of project choice, but we do not discuss it in details due to significant differences from our model.
A line of work~\cite{lewis2012theory,xiao2022information} study a delegated search problem, especially a dynamic version by \cite{rahmani2016dynamics}, but the players bear the cost of search for solutions in their model, whereas the solutions are exogenous to the mechanism in our model.

\paragraph{Stackelberg games}
Our problem can be viewed as an online learning version of repeated Stackelberg game \cite{stackelberg1934marktform,marecki2012playing,bai2021sample,lauffer2022no,zhao2023online}.
A common objective in this area of work is to minimize a Stackelberg regret, \ie difference to the optimal policy that knows the leader's optimal action in hindsight, and the above works aim to minimize the cumulative Stackelberg regret of a leader, assuming that a follower best responds at each round.
Especially, our model of strategic agent belongs to the growing area of learning in games with strategic agent
\cite{birmpas2020optimally,haghtalab2022learning,zhao2023online}.
More precisely, \cite{birmpas2020optimally} study how the follower can efficiently deceive the leader by misreporting his valuation.
\cite{haghtalab2022learning} proposes a generic delaying technique to deal with a strategic agent, and proposes several applications to strategic classification, repeated posted price mechanism (henceforth RPPM), and Stackelberg security game.
Indeed, our model resembles RPPM of \cite{kleinberg2003posted,amin2013learning,babaioff2017posting}.
However, RPPM restricts the buyer and the seller's utility to be linearly negatively correlated, but our model accommodates any kind of correlation.
In addition, our agent has multiple solutions to choose from compared to only accept/reject of RPPM, and thus is technically more challenging to predict/analyze the agent's strategic behavior.


\section{Problem Setup}\label{sec:model}
In a \emph{repeated delegated choice} problem, there is a principal and an agent.
The agent is equipped with a set of solutions $A = \{a_0, a_1, \ldots, a_K\}$ where $K$ denotes the cardinality of the set of possible solutions\footnote{We do not restrict the number of solutions to be finite, or constant with respect to T.}, and $a_0$ denotes the null solution $\perp$ which means that the agent submits nothing.
At each round $t \in [T]$, solution $a$ incurs a nonnegative random utility for the principal and the agents.
Denote by $X_a^{(t)}$ the utility random variable (r.v.) of the principal selecting the solution $a$ and $Y^{(t)}_{a}$ the random utility of the agent given solution $a$, both of which has support in $\Omega := [0,1]$. 
The random vector $(X_a^{(t)}, Y_a^{(t)})$ is independent and identically distributed (i.i.d.) for $t \in [T]$.
Importantly, the agent can \emph{access} the ex-post utility of all the solutions $a \in A$ at each round, but the principal cannot.
The agent is equipped with a discounting factor $\gamma \in (0,1)$, \ie he discounts the utility at round $t$ by a factor of $\gamma^{t-1}$.
That is, the agent's true utility for solution $a$  at round $t$ is $\gamma^{t-1}Y_a^{(t)}$. 
This assumption on agent regret is common in studies concerning strategic agents \cite{amin2013learning,haghtalab2022learning}. It is shown in \cite{amin2013learning} that in the repeated posted-price mechanism problem, a sublinear regret can not be achieved for a non-discounting strategic agent. This is also the case in our problem, as a non-discounting agent might have the incentive to hide a solution that is worse than another solution in terms of agent utility but better for the principal. 
We further explore this in Appendix K.
Define $\tgamma = 1/(1-\gamma)$.
Given a mechanism $M$, at each round $t \in [T]$, the agent chooses a (possibly random) subset of solutions $S^{(t)} \subset A$, and submits them to the principal.
We write $2^X$ to denote the power set of a set $X$, and $\Delta(X)$ for a simplex over $X$.
Thus, the agent's action belongs to $S^{(t)} \in \Delta(2^{A})$.

\paragraph{History, mechanism, and agent's policy}
At each round $t \in [T]$, the mechanism determines which solutions to commit given the agent's action $S^{(t)}$.
This choice is based on the history available up to round $t$, formally defined by $H_t := \cup_{l=1}^{t-1}(S^{(l)},a^{(l)})$, where $a^{(l)}$ denotes the solution selected at round $l$.
We define $\cH:= \cup_{t \ge 1} (2^A, A)$ to be the set of all possible histories of the game, \ie each $H_t$ is a subset of $\cH$.
Formally, the mechanism $M:\cH \times 2^A \mapsto A$ specifies which solutions to select at each round $t$ given the history $H_t \in \cH$ and the agent's submission.
Importantly, the mechanism is only able to choose an action among the actually submitted solutions by the agent.
Also, the principal \emph{commits} to a mechanism before the game starts.

Let $\cM$ be the set of all possible mechanisms.
Correspondingly, the agent's policy $P:\cH \times \cM \mapsto \Delta(2^A)$ is a function that takes the mechanism announced by the principal and the history sequence $H_t$ and decides an (possibly randomized) action.
Let $\cP$ be the set of all possible agent policies.
We write $s_t$ to denote the solution eventually selected at round $t$ by the mechanism.
Note that $s_t$ can be null, \ie $s_t = \perp$, if the principal declines to accept any proposed solution.
In this case, both the agent's and the principal's utilities are zero.
We write $X_{M,P}^{(t)}$ and $Y_{M,P}^{(t)}$ to denote the principal's and agent's utility at round $t$ under $M$ and $P$.
Define 
$\Phi_{M,P} = \mathbb{E}[X_{M,P}^{(t)}]$ and  $\Psi_{M,P} = \mathbb{E}[Y_{M,P}^{(t)}]$
to denote the expected utility of the principal and the agent, respectively.

\paragraph{Mechanism description}
Overall, the interaction between the principal and the agent proceeds as follows:
\begin{enumerate}[i]
    \item The principal commits to a mechanism.
    \item At each round, agent observes the realized solutions and their utility.
    \item Agent (possibly strategically) proposes solutions.
    \item Principal observes the proposed solutions and corresponding utility, and determines the final outcome with respect to the committed mechanism.
    \item Steps ii-iv are repeated.
\end{enumerate}

\paragraph{Single-proposal mechanism}
We mainly deal with the following specific type of mechanism, inspired by \cite{kleinberg2018delegated}.
\begin{definition}[Single-proposal mechanism]\label{def:single}
	In a single proposal mechanism $M$, at each round $t$, the principal announces an eligible set $E^{(t)} \subset \Omega^2$, and the agent submits only a single solution $a$.
	If $(X_a^{(t)}, Y_a^{(t)}) \in E^{(t)}$, then the principal accepts the solution, otherwise, she selects nothing.
\end{definition}
We further say that a mechanism is threshold-based, if its eligible set only puts a (possibly strict) lower bound on the principal's utility. We define $E_\tau = \{a: X_a \ge \tau\}$ and $E_\tau^{>} = \{a: X_a > \tau\}$ to represent threshold-based eligible sets for a threshold $\tau$.
Given a single proposal mechanism, we write $x_a^{(t)}$ and $y_a^{(t)}$ to denote the eventual utility of the principal and the agent at round $t$ when the agent proposes solution $a$, \ie which reflects the principal's decision.

Notably, we provide a revelation principle style of result which states that any mechanism can be reduced to a single-proposal mechanism.
\begin{theorem}\label{thm:revelation}
    Given any mechanism $M$ and any agent's policy $P$, there exists a single-proposal mechanism $M'$ and corresponding deterministic agent's policy $P'$ such that $\Phi_{M,P} \le \Phi_{M',P'}$ and $\Psi_{M,P} \le \Psi_{M',P'}$.
\end{theorem}
Thanks to the reduction above, we can essentially focus on the single-proposal mechanism, and the agent only needs to determine which solution to submit at each round.
Thus, unless specified explicitly, we now focus on the single-proposal mechanism.
Note that the reduction from any deterministic mechanism with deterministic policy follows from a variant of the proof of the standard revelation principle~\citep{nisan2007algorithmic}.
For randomized policy, we can reduce it to a deterministic policy by sequentially derandomizing each round's random events in a backward manner.

\paragraph{Approximately best response and Stackelberg regret}
Our construction of a mechanism against a strategic agent requires a notion of approximate best response of the agent, defined as follows.
\begin{definition}[$\eps$-best response]
	Given a mechanism $M$ and history $H_t$, let $A_E$ be a union of the set of eligible solutions given eligible set $E$ and the null outcome $\perp$.
	Then, the $\eps$-best response at round $t$ for eligible set $E$ is defined by
	\begin{align*}
		\br_\eps^{(t)}(E) = \{a \in A_E: y_a^{(t)}\ge y_{a'}^{(t)}- \eps, \forall a' \in A\}.
	\end{align*}
	If $\eps = 0$, we simply say best response and denote by $\br^{(t)}$.
\end{definition}
Whenever there are multiple solutions as best response, we assume that a myopic agent plays in favor of the principal, \ie submits the solution that maximizes the principal's utility.

Fundamentally, the dynamics of the single-proposal mechanism belongs to a repeated Stackelberg game in which the principal moves first by announcing an eligible set, and then the agent follows by proposing solutions, at each round.
In repeated Stackelberg games (possibly with strategic agent), typical objective is to minimize a cumulative regret compared to the case when the mechanism knows the optimal eligible set in hindsight, and the agent \emph{myopically} responds to the principal's move.
In our setting, this benchmark boils down to the case under which the mechanism knows the distribution of $X_a^{(t)}$ and $Y_a^{(t)}$ in hindsight, while the agent best responds to the principal's eligible set at each round.
In this case, the optimal principal's utility can be written as,
\begin{align}\label{eq:opt_single}
	\opt = \max_{E \subset \Omega^2}\Ex{x_{\br^{(t)}(E)}^{(t)}}.
\end{align}
Thus, Stackelberg is defined as follows.
\begin{definition}[Stackelberg regret]
	Given a mechanism $M$ and agent's policy $P$, suppose that the agent submits solution $a_t$ at each round $t$.
	Then, Stackelberg regret is defined by
	\begin{align*}
		\reg_{M,P}(T) = T\cdot \opt - \sum_{t=1}^T x_{a_t}^{(t)}.
	\end{align*}
\end{definition}
Furthermore, we define a worst-case Stackelberg by maximizing over the agent's policy, $\wreg_M(T) = \max_{P \in \cP}\reg_{M,P}(T)$.
Let $\cP_\eps$ be a family of policy under which the agent always plays $\eps$-best response.
Then, we define $\wreg_M(T,\eps) = \max_{P \in \cP_\eps}\reg_{M,P}(T)$.
If the agent is myopic, we abuse $\reg_M(T) = \wreg_M(T,0)$ to denote its worst-case Stackelberg regret.

\section{Deterministic Setting}\label{sec:deter}
We start with a simple setting in which the agent is myopic and the utility is deterministic.
In this case, the principal needs to learn the optimal eligible set, without regard to the agent's strategy.
In this case, for notational simplicity, we drop the superscript $(t)$ since the utility remains the same across the rounds. 
Our main result with myopic agent is presented as follows.
\begin{theorem}\label{thm:deter_myop}
    There exists a mechanism with $\reg(T) = O(\min(K,\log\log T))$ against myopic agent.
\end{theorem}
The proof is based on two algorithms, one of which relies on a novel connection between our problem and the repeated posted-price mechanism problem (henceforth RPPM) by~\cite{kleinberg2003posted}, and the other is a simple algorithm that iteratively finds a (slightly) better solution.
In the former, we mainly construct a reduction from our problem to RPPM, and thus recovers the regret bound $\log \log T$ of~\cite{kleinberg2003posted}.\footnote{Note that any state-of-the-art result can be carried over to our problem's regret bound, due to our reduction.}
In the latter algorithm with regret bound $O(K)$, the algorithm iteratively updates the eligible set so that it \emph{excludes} at least one suboptimal solution at each round, until there's no eligible solution.
Formal definition of RPPM, pseudocode of the algorithms, and the proof can be found in the appendix.

This result implies an intimate connection between our problem and RPPM, however, we observe that this does not hold beyond this simplistic setting.
In fact, the utilities are always linearly negatively correlated in RPPM, on the other hand, in our setting they can be arbitrarily correlated.
Moreover, the agent has multiple solutions to choose compared to only two actions of accept or reject in RPPM.

\subsection{Strategic Agent}\label{sec:deter_strategic}
Next, we consider a more challenging scenario in which the agent tries to strategize over the rounds.
Since we cannot assume that the agent will truthfully best respond to the mechanism at each iteration, instead,
he possibly tries to deceive the mechanism by untruthfully submitting a solution.
Thus, we need to design a mechanism that is robust to the strategic behavior.
This may indeed be plausible in practice, for instance in our online labor market example, the platform may try to deceive the task requester to not strain highly qualified workers.
Intuitively, this will be especially true when the platform does not have a large number of workers.

Mainly, we characterize the regret upper bound with respect to two types of assumptions.
The first version of the results relies on a relatively simple assumption such that the agent's utility is uniformly bounded below by some constant.
The latter depends on the Lipschitz continuity of the utility across the solutions and the density of solutions in the utility space.
For each setting, we provide regret upper bounds and matching lower bounds.
This justifies the necessity of the assumptions imposed, thereby characterizing the regimes in which the principal can attain large utility.

Before presenting the results, we introduce a notion of delayed mechanism, which will be useful in dealing with strategic agent.
Formally, we say that mechanism $M$ is $D$-delayed, if at each round $t$, it uses $H_{\max(1,t-D)}$ to decide its eligible set $E_t$.
Delayed mechanism effectively restricts the strategic agent's behavior, as follows.
\begin{lemma}[\cite{haghtalab2022learning}]
	\label{thm:haghtalab1}
    Given $\gamma \in (0,1)$, if we set $D = \lceil T_\gamma \log(T_\gamma / \eps)\rceil$, then $D$-delayed mechanism $M$ satisfies $\wreg_M(T) \le \wreg_M(T, \eps)$.
\end{lemma}
Intuitively, if $D$ gets larger enough, the agent with discounted utility is less incentivized to deviate from the best response at each round since the discounted utility after $D$ rounds may not be enough to make up for the loss of $\eps$ utility in the current turn. 

\paragraph{Uniformly bounded agent utility}
Formally, we first assume that $Y_a > \ymin$ for any $a \in A$ and the principal is also aware of this lower bound.
Our regret bound will accordingly be parameterized with respect to $\ymin$.
This assumption is plausible since in our online labor market example, the task requester and the worker are typically contracted to pay an intermediary fee to the platform and thus constitute a reasonable amount of minimum payoff to the agent.
The existence of such a minimum utility  effectively allows us to compute the necessary delay to make
the agent approximately myopic, thanks to Theorem~\ref{thm:haghtalab1}.

Leveraging the minimum utility of the agent, we consider a variant of the algorithm used in the myopic deterministic setting, by introducing a delay in reacting to the agent's feedback.
Then, we can obtain the following regret bound.
\begin{theorem}\label{thm:alg1}
    There exists an algorithm with $\wreg(T) = O(K\tgamma \log{\frac{\tgamma}{\ymin}})$ against $\gamma$-discounting strategic agent.
\end{theorem}
Essentially, the delay introduced in the algorithm induces the agent to behave \emph{restrictively strategic},
and we can effectively bound the regret to be constant, assuming the other parameters are constants.
Note, however, that the regret bound linearly depends on the number of agent's solutions $K$.
Obviously, if $K$ tends to be large in some cases, our regret guarantee here is doomed to be pessimistic.
This is indeed plausible in practice, since the agent may have growing number of solutions with respect to $T$, especially for online platforms.



This limitation can be handled by shrinking the eligible sets more in an aggressive manner, instead of sequentially seeking the next-best solutions.
Then, the linear dependency on $K$ can further be wiped out as follows.
\begin{theorem}\label{thm:alg2_regret}
    There exists an algorithm with $\wreg(T) = O(\tgamma \log{\frac{\tgamma}{\ymin}} + \log{T})$ against $\gamma$-discounting strategic agent.
\end{theorem}
Note that the regret no longer depends on the number of solutions $K$, but instead on $\log T$.
Our algorithm keeps shrinking the eligible set until it concludes that the truly optimal solution lies within at most $1/T$ to the currently best solution.
Afterward, the regret is at most $1/T \cdot T$, thus does not affect the overall regret upper bound.
Intuitively, to remove the dependence on the number of solutions, such a logarithmic burden on $T$ is essential to guarantee that our eventual solution is correct up to $O(1/T)$ distance.

We further note that both the regret bounds of Theorem~\ref{thm:alg1} and~\ref{thm:alg2_regret} have a logarithmic dependency on $1/\ymin$.
Assuming that $T_\gamma = 1/(1-\gamma) = O(1)$, this regret bound yields a sublinear regret upper bound if $\ymin = e^{-o(T)}$, but becomes detrimental the other way around.
Interestingly, however, we show that this dependency is necessary to obtain sublinear regret for any algorithm, by formally proving that no algorithm can achieve sublinear regret if $\ymin \le e^{-T}$ against the strategic agent. 




\begin{theorem} \label{thm:linear_regret}
    If $\ymin \leq e^{-T}$, then any algorithm has $\wreg(T) = \Omega(T)$ against $\gamma$-discounting strategic agent.
\end{theorem}

Thus, the principal suffers a large amount of regret by delegating to the agent whose utility tends to be very small.

\paragraph{Lipschitz utility with dense solutions}
Next, we consider the case where the principal is not aware of any lower bound on $\ymin$, or such a lower bound does not exist.
To cope with this lack of information on $\ymin$, we assume that there is no significant disparity in the utility between two close solutions for both the principal and the agent, and the solutions are densely spread in the utility space.
Under these assumptions, we provide an algorithm to find a near optimal solution.



These assumptions are formally presented as follows.
\begin{assumption}[$d$-dense]
        Let $d_X(a,b) = |X_a-X_b|$.
	A problem instance is $d$-dense for some $d > 0$ if for any two solutions $a,b \in A$, either of the following is satisfied: (i) $d_X(a,b) \le d$ or (ii) $d_X(a,b) > d$ and there exists another solution $c$ such that $d_X(a,c) \le d$ and $d_X(b, c) \le d_X(a, b)$.
\end{assumption}


\begin{assumption}[$L_1,L_2$-Lipschitz continuity]
    There exists absolute constants $L_1,L_2 > 0$ such that for any $a,b \in A$, we have $L_1 \cdot d_X(a,b) \le d_Y(a,b) \le L_2 \cdot d_X(a,b)$, where $d_X(a,b) = |X_a - X_b|$ and $d_Y(a,b) = |Y_a - Y_b|$.
\end{assumption}





Our assumption of densely spread solutions is innocuous since the solutions will be packed more in a~compact manner as the number of solutions grow.
Otherwise, if the number of solutions is relatively small, then our results on the bounded agent's utility would kick-in, and thus one may recover the sublinear regret.
The Lipschitz continuity assumption is often valid, as it is observed that when the agent's utility for two solutions is similar, the principal's utility follows suit, and vice versa.
For instance, if all the solutions lie in $y = 1-x$, then the Lipschitz condition holds with $(L_1,L_2)$ being $(1,1+\eps)$ or $(1-\eps,1)$ for any choice of $\eps \ge 0$.

Further, we assume that Lipschitz parameters $L_1$ and $L_2$ tend to be close to each other, precisely, $L_2 \geq L_1 > \frac{3}{4} L_2$. Indeed, if there exists a significant difference between $L_1$ and $L_2$, the Lipschitz assumption fails to effectively impose any restrictions.


Leveraging these assumptions, we propose a new algorithm.
Since we lack precise information on the required delay to ensure the submission of a solution, we cannot compel the agent to be approximately myopic.
We propose a modified version of the algorithm used above which effectively leverages the assumptions above to explore superior solutions.


\begin{algorithm}
    \SetAlgoLined
    \While{any solution has not been received}{
        Announce $E_0$.\\
    }
    Let $a_0$ be the proposed solution. \\
    $\alpha \gets L_1 - \frac{3}{4}L_2$, $l \gets X_{a_0}$, $y \gets Y_{a_0}$, $r \gets \min\{1, l + \frac{y}{L_1}\}$, $\eps \gets 4\alpha d$, $D \gets \tgamma \log{\frac{\tgamma}{\eps}}$; \\
    Announce $E_0$ for $D$ rounds. \\
    \While{$r-l > 4d$}
    {
        $\tau \gets \frac{l+r}{2}$; \\
        Announce $E^{>}_{\tau}$. \\
        \leIf{solution $a$ is proposed by the agent}{
            $l \gets X_a$, $y \gets Y_a$
        }
        {
            $r \gets \tau$
        }
        Announce $E_l$ for $D$ rounds. \\
        $r \gets \min\{r, l + \frac{y}{L_1}\}$; \\
    }
    Announce $E_l$ for remaining rounds.
\NoCaptionOfAlgo
\caption{Algorithm 1: \textsc{DelayedProgeressiveSearch}}\label{alg:3}
\end{algorithm}



\begin{theorem}\label{thm:lips-lns-reg}
    If $\alpha := L_1 - \frac{3}{4}L_2 > 0 $,
    then Algorithm~\ref{alg:3} has
    $\wreg(T) = O(\tgamma \log{\frac{\tgamma}{\alpha}} + \log{\frac{1}{d}} + d T)$ against $\gamma$-discounting strategic agent.
\end{theorem}
In Algorithm~\ref{alg:3}, we maintain an interval,
denoted as $[l, r]$, which encompasses the optimal solution.
It is guaranteed that at every round, a solution $a$ exists such that $X_{a} = l$.
According to the Lipschitz continuity assumption, we can place an upper bound on $r$,
signifying that the optimal solution $a^*$ cannot be significantly distant from $a$.
This is because when the difference between $X_{a^*}$ and $X_{a}$ becomes large,
it is expected that $Y_{a^*} - Y_{a}$ will also be substantial.
This is not possible since $Y_{a^*}$ is non-negative, and cannot be significantly greater than $Y_{a}$,
otherwise, it would have been proposed by the agent in earlier rounds.

With the bounded value of $r$, our objective is to determine if there exists a solution within the right half of the interval. By considering the line $x = \frac{l+r}{2}$, the $d$-dense assumption implies that if a solution exists in the right half, there must be a solution with the principal's utility ranging from $\frac{l+r}{2}$ to $d + \frac{l+r}{2}$. Utilizing the Lipschitz continuity along with the condition on its parameters, we can find a lower bound on the agent's utility within that interval. Consequently, we can introduce an appropriate delay to compel the agent to propose a solution from the right half if it exists. As a result, the algorithm can determine the presence or absence of a solution in the right half and subsequently shrink the interval accordingly. By continuing this procedure, the interval gradually converges toward the optimal solution.

Furthermore, our regret upper bound essentially decomposes $\tgamma$ from $T$, and thus the effect of discount factor is decoupled from the linear dependency of $dT$.
The linear dependency on $dT$ may look pessimistic at first glance, but we reveal that this dependency is indeed optimal, as formally presented as follows. 
\begin{theorem} \label{thm:dense-lw}
    There exists a $d$-dense problem instance such that any algorithm suffers $\wreg(T) = \Omega(dT)$ against $\gamma$-discounting strategic agent.
\end{theorem}
Its proof easily follows from the proof of Theorem~\ref{thm:linear_regret}.
Thus, this demonstrates the fundamental inevitability of the term $dT$.
It's worth noting that whenever $d$ is subconstant, \eg $d = T^{-c}$ for $c>0$, then our regret upper bound in Theorem~\ref{thm:lips-lns-reg} implies a sublinear regret.

In our online labor market example, since the number of workers in a platform usually grows with respect to the time horizon, their intrinsic qualities might lie more compactly in the utility space as time flows.
For instance, if there are $T^{\eps}$ workers having uniformly distributed utility in a compact utility space for some $\eps \in (0,1)$, their utility will be $O(T^{-\eps})$-densely spread, which would yield $dT = T^{1-\eps} = o(T)$ regret bound.
In words, the requester needs to delegate to a platform with a large number of solutions, \ie delegating to \emph{big business matter}.
Conversely, the platform should maintain more workers to attract requesters, \ie \emph{economy of scale} works, however, a trade-off arises since the platform's gains from strategizing would be limited then.

\section{Stochastic Setting}\label{sec:stoc}
For the stochastic setting, previous algorithms no longer work as the ex-post optimal solution varies across the rounds.
Thus, the objective of the principal here will not be to find the largest threshold $\tau$ to exclude any ex-post suboptimal solutions, but to balance the probability that the agent possesses a solution that belongs to the eligible set and the corresponding utility of the agent's best response therein.
Recall that this phenomenon is already well-captured in our benchmark~\eqref{eq:opt_single} and the corresponding notion of Stackelberg regret.
Given the differences, the principal faces an additional challenge of handling the random noise in the reward, and needs to find the best threshold to balance the trade-off presented above.

To cope with it, we reduce our problem to a stochastic bandit problem which is standard in the literature~\citep{kleinberg2003posted,amin2013learning, haghtalab2022learning}.
In the stochastic multi-armed bandit problem, a principal has a set of $Q$ arms, indexed by $i \in [Q]$, and needs to decide which arm to pull at each round given the time horizon $T$.
Each arm $i$ is equipped with a reward distribution $D_i$ with support $[0,1]$.
Given $\mu_i = \Exu{r\sim D_i}{r}$, the principal's objective is to minimize expected regret defined by $\reg(T) = T\cdot \max_{i \in [Q]}\Exu{r \sim D_i}{r} - \sum_{t=1}^T \Ex{r_t}$, where $r_t$ denotes the random reward of the arm selected at round $t$ by the principal.

We first discretize the space of principal's utility into the set of $i/Q$ for $i \in [Q]$ for some carefully chosen parameter $Q$.
Each element $i/Q$ corresponds to a single arm, which represents the threshold $\tau$ that the principal can commit to at each round.
By pulling the arm $i$, the principal is essentially announcing an eligible set of $E_{\tau_i} = \Ind{a\ |\ X_{a} \ge i/Q}$.
Namely, the principal aims to find the best eligible set among the set of $E_{\tau_i}$ for $i \in [Q]$.
If the discretization is dense enough with respect to the problem parameters, the regret bound here would imply a reasonable regret bound for our original problem. We define $f(\tau) = \Ex{x^{(t)}_{\br^{(t)}(E_\tau)}}$ as the expected principal utility for using threshold $\tau$.
Note that if the agent best responds, the expected utility from pulling arm $i$ becomes $f(i/Q)$.

We assume that $f(\tau)$ achieves its maximum for a unique $\tau^* \in (0,1)$ with $f''(\tau^*) < 0$, which is common in the literature \citep{kleinberg2003posted,amin2013learning, haghtalab2022learning}.
Now, we can simply use the well-known UCB upon the discretization, and obtain the following results against the myopic agent.
\begin{theorem}\label{thm:stoc_myop}
	If the agent is myopic, running UCB1 with discretization by $Q = (\frac{T}{\log{T}})^{1/4}$ has $\reg(T) = O(\sqrt{T\log{T}})$.
\end{theorem}
Its analysis is a simple variant of \cite{kleinberg2003posted}, but we provide the entire proof to make paper self-contained.

Next, to deal with the strategic behavior of the agent, we again exploit the concept of delay to restrict the agent to be approximately best responding with a suitable choice of parameters.
In addition, however, we cannot simply expect that the outcome of pulling a single arm, \ie a specific eligible set, follows some stochastic distributions since the agent may strategically deviate from the best response at hand.
To cope with this additional challenge, we use the foundation of perturbed bandit instance by \cite{haghtalab2022learning}.

Since the stochastic setting is a generalization of the deterministic setting, it is obvious that we need a reasonable set of assumptions to obtain positive results.
Similar to the deterministic setting, we first assume that there exists a value $y_{min}>0$ such that for any realization of the agent's solutions, his utility for each solution is strictly greater than $y_{min}$. 
Secondly, we assume that the problem instance satisfies the following assumption, which is a stochastic version of Lipschitz continuity in the deterministic setting. 
\begin{definition}[Stochastic Lipschitz continuity]
	Under the stochastic setting, we say that the problem instance is stochastically Lipschitz-continuous with parameter $L_1 > 0$ if the ex-post utilities of the solutions are correlated in a sense that for any $a,b \in A-\{a_0\}$, we have $L_1 \cdot d_X(a,b) \le d_Y(a,b)$.
\end{definition}
Finally, our main result can be presented as follows.
\begin{theorem}\label{thm:stoc_stra}
	Under the two assumptions presented above, there exists an algorithm that has $\wreg(T)$ of
 \begin{align*}
      O\parans{\sqrt{T\log T} + T_\gamma \log\left(T_\gamma \max\left(\frac{T}{L_1},\frac{1}{\ymin}\right)\right)\log{T}},
 \end{align*}
 with $\gamma$-discouting strategic agent.
\end{theorem}
Our proof essentially relies on a construction of proper random perturbation interval, followed by the regret analysis of delayed version of successive elimination algorithm by \cite{haghtalab2022learning} and \cite{lancewicki2021stochastic}.
The technical subtlety lies on introducing a proper random perturbation interval to convert it to the perturbed bandit instance.
The latter term including $\tgamma, L_1$ and $\ymin$ incurs due to the strategic behavior of the agent.
Still, this only contributes $\log T$ amount of regret once all these parameters are constants, which is dominated by the former term of $\sqrt{T \log T}$.

\section{Conclusion}
We study a novel \emph{repeated delegated choice} problem.
This is the first to study the online learning variant of delegated choice problem by \cite{armstrong2010model,kleinberg2018delegated,hajiaghayi2023multi}.
We thoroughly investigate two problem dimensions regarding whether the agent strategizes over the rounds or not, and whether the utility is stochastic or deterministic.
We obtain several regret upper bounds for each problem setting, along with corresponding lower bounds that complement the hardness of the problems and some assumptions therein.
Our analysis mainly characterizes the conditions of problem instances on which the principal can efficiently learn to delegate compared to the case when she knows the optimal delegation mechanism in hindsight, thereby providing fruitful insights in the principal's decision-making in delegation process.

\bibliographystyle{ACM-Reference-Format}
\bibliography{ref}

\newpage
\appendix
\section{More discussions on motivation}\label{app:moti}
As shown by the online labor market exampled presented in the introduction, repeated delegation between two parties with misaligned utility and information asymmetry happens in many real-world scenarios.
For example, consider a university (principal) bears the cost of hiring faculty members while delegating the role of selection to the department (agent).
The university and the department may have different goals in hiring faculty members, while the university may constrain feasible candidates of faculties in advance.
In addition, a content platform like Youtube (principal) decides which contents to recommend to incoming users while requiring Youtubers (agent) to register contents.
This can essentially be viewed as a process of delegation, since although Youtube may have registered the contents by itself, it delegates the role of registering contents to Youtubers due to a lack of time, money, or expertise.
Finally, one may consider an e-commerce (principal) decides which items to recommend to incoming buyers.
Typically, there exists a number of sellers (agent) who contract with the e-commerce, and register their items in the platform.
Thus, e-commerce delegates the role of actually creating items and deciding which ones to sell.
Note that in the examples above, different from the online labor market example we provided in the introduction, the platforms typically act as a principal and the individuals can be viewed as an agent.

\section{Connection to repeated posted pricing problem} \label{app:connection-to-rppm}

We found an intimate connection between the repeated posted price mechanism (RPPM) and our problem.
For example, offering a price in RPPM can be viewed as announcing a thresholded eligible set in our problem.
The acceptance can be regarded as a submission of eligible solution, and rejection as the agent declined to submit any solution.
However, the correspondence is not exact, since in RPPM there is only one value for the buyer w.r.t. to the item,
but in our problem there are multiple values w.r.t. multiple solutions that the agent might possess.
Therefore, while the only latent value that the principal needs to learn is the buyer's value for the item in RPPM,
in our setup,
the principal needs to learn the value of multiple solutions.
Finally, the principal does not know how many solutions the agent possesses in advance.
Besides, different from the fact that the principal's utility (revenue = price) is exactly negatively correlated to the agent's utility (valuation - price),
in our setting,
it can be arbitrarily correlated.

To provide a simple example of our analysis,
given two solutions that yields principal's ($x$) and agent's ($y$) utility of: $(x, y)=$(i) $(0.5, 0.5)$ and (ii) $(0.75, 0.75)$, there is no guarantee that the agent would pick (ii), although it strictly dominate (ii). This is because the agent may think suggesting (ii) would increase the principal's expectation on the further solutions, so that it prevents possibly inferior solutions to be eligible in the future. This becomes more complicated when there are multiple Pareto-optimal solutions w.r.t. the principal's and agent's utility.
Despite these complications, our analysis, however, reveals that one can construct efficient algorithms by incorporating existing techniques and some new ideas on how to deal with agent (especially in strategic case).


\section{Proof of Theorem~\ref{thm:revelation}}
\begin{proof}
    To begin with, suppose first that $M$ is a deterministic mechanism.
    Then, there exists a deterministic sequence $(a_t)_{t \in [T]} \in (A\cup\{\perp\})^T$ representing the (possibly null) arms selected at each round.
    Since the agent is fully capable of computing the exact payoff, 
    Denote the corresponding agent's actions sequence by $S^{(t)}_{t \in [T]}$.
    Without loss of generality, suppose that the realization of the $X_a^{(t)}$ for the problem instances with $M$ and $M'$ remains the same.
    Given $H_t$, let $M'$ be a single-proposal mechanism that announces an eligible set which consists of the selected solution (possibly null) by mechanism $M$ at round $t$.
    Consider the agent's policy $P'$ such that for each $t \in [T]$, given $H_t$ and $X_a^{(t)}$ for $a \in A$, he submits $a_t$ if $a_t \neq \perp$, and otherwise he submits nothing.
    Then, it follows from the above construction that $X_{M,P}^{(t)} = X_{M',P'}^{(t)}$ and $Y_{M,P}^{(t)} = Y_{M',P'}^{(t)}$, and thus we essentially have $\Phi_{M,P} = \Phi_{M',P'}$ and $\Psi_{M,P} = \Psi_{M',P'}$.

    Suppose that rewards of solutions realized to be $\om(t) = (\om_1(t),\ldots, \om_K(t))$ for each $t \in [T]$.
    Now consider a randomized policy of agent $P$ and (possibly) randomized mechanism $M$, where $P$ maximizes the agent's expected payoff.
    Let $\Sigma_P$ be the set of all possible random bits if $P$, and similarly define $\Sigma_M$ for $M$.
    For each $\sigma = (\sigma_P, \sigma_M) \in \Sigma = \Sigma_P \times \Sigma_M$, we can essentially compute the agent's expected payoff $y_{M,P(\sigma)}(\omega)$.
    Let $\sigma^*(\omega) = \argmax_{\sigma \in \Sigma}y_{M(\sigma_M),P(\sigma_P)}(\omega(t))$.
    Let $P'$ be the agent's deterministic policy which plays $P(\sigma^*(\omega(t)))$ against $M$ for each realization $\omega(t)$.
    Then, due to our construction, $P'$ yields a larger (or equal) agent's utility than $P$.
    If this yields a strictly larger utility, then the agent should have played $P'$, which is a contradiction.
    Thus, the agent's policy can be made deterministic, and by the argument presented above, we finish the proof.

\end{proof}

\section{Proof of Theorem~\ref{thm:deter_myop}}\label{sec:deter_myop}
\begin{proof}
    As we discussed in the main paper, the proof is essentially based on two observations (i) reduction to RPPM, and (ii) iterative algorithm.
    To this end, we first present the formal definition of RPPM.
    \begin{definition}[Repeated Posted Price Mechanism]
	In a \emph{repeated posted price mechanism} problem, given $T$ rounds, a buyer with an i.i.d. valuation $v_t$ from an unknown distribution supported on $[0,1]$ arrives at each round, and decides whether to buy an item.
	The seller posts a price $p_t$ on the item, and the buyer takes the item if $v_t \ge p_t$, and leaves otherwise.
	The seller's regret is defined by $\reg(T)  = T\cdot\max_{p \in [0,1]}p\cdot \Pr{v \ge p} - \Ex{\sum_{t=1}^T p_t\Ind{p_t \le v_t}}$, \ie the difference compared to the optimal revenue in hindsight with a fixed price.
    \end{definition}
    Now we prove the following claim.
    \begin{claim}
        Suppose that the agent is myopic. 
        Given a mechanism for the RPPM problem that achieves regret of at most $R(T)$ on any instance with $T$ rounds, we can construct a mechanism $M$ for the repeated delegated choice achieving the same regret bound $R(T)$ for all instances with $T$ rounds. 
    \end{claim}
    \begin{proof}[Proof of the claim]
        
        Let $M$ be a mechanism for the RPPM problem. We construct a threshold-based mechanism $M'$ for the repeated delegated choice problem that simulates $M$. At each round, the mechanism $M$ chooses a price based on the history of the buyer's responses. In the RPPM problem, this history consists of the prices used at each round and whether the buyer bought the item or not. For $M'$, we consider the history to be the threshold $\tau$ used in each round which defines the eligible set $E_\tau$ and whether the agent proposes an eligible solution or not. Now, at each round, we can consider the history for $M'$ to be a history for $M$. If $M$ would propose price $p$ with the given history, $M'$ announces a threshold $\tau=p$. 

        Now, consider an instance $I'$ of the repeated delegated choice problem  with $T$ rounds. In the deterministic case, each random variable $X_a^{(t)}$ takes only a single value, and so we abuse this notation to refer to the value rather than the random variable. Let $X^{(t)}=\max_{a\in A} X_a^{(t)}$ denote the maximum value of the principal's utility among the agent's solutions. We define instance $I$ of the RPPM problem, where the buyer's value is $X^{(t)}$ and there are $T$ rounds. We show that the regret of $M'$ on $I'$ is at most that of $M$ on $I$. We prove inductively that at each round, the price $M$ proposes given $I$ is the same as the threshold announced by $M'$ for $I'$, the agent and the buyer's responses are the same which leads to the same history, and the principal utility achieved by $M'$ is at least as high as the seller utility in $M$. In the first round, there is no history and in the other rounds, the history will be the same by induction. So by definition, $M'$ uses the price $\tau$ used in $M$ as the threshold. In $I$, the buyer buys if and only if $X^{(t)}$ is at least $\tau$. Similarly, in $I'$ the agent proposes an eligible solution with principal utility at least $\tau$ if and only if such a solution exists as this would result in a positive utility for the agent as opposed to the zero utility for submitting an ineligible solution or no solution. So, the agent and buyer responses are the same. In addition, if $X^{(t)}\geq \tau$ the seller's utility is the price $\tau$, and the principal's utility is at least $\tau$ as an eligible solution is submitted. If $X^{(t)}\leq \tau$, the buyer does not buy and the agent can not propose an eligible solution, so both the seller and the principal's utility will be $0$. So, the principal's utility when using mechanism $M'$ in each round is at least as high as the seller's using $M$.

        The regret in both cases is defined as the difference between the utility achieved and the optimal utility, which is $T\cdot X^{(t)}$ in both cases. Therefore, the regret of $M'$ on $I'$ is at most the regret of $M$ on $I$. So, if $M$ has regret at most $R(T)$ on instances with $T$ rounds, $M'$ will also achieve the same regret bound.

    \end{proof}

    Now, the regret upper bound of $\log \log(T)$ can directly be obtained from the following theorem. 
    \begin{theorem}[\cite{kleinberg2003posted}]
        There exists an algorithm achieving regret $O(\log \log T)$ for the RPPM problem in the deterministic myopic setting.
    \end{theorem}

    \begin{algorithm}
        \SetAlgoLined
        $\tau = 0$, $f=0$\\
        \While{$f \neq 1$}
        {
            Announce $E^{>}_{\tau}$ for a single round. \\
            \leIf{arm $a$ is proposed by the agent}{
                $\tau \gets X_a$
            }
            {break}
        }
        Announce $E_{\tau}$ for remaining rounds.
    \caption{\textsc{IterativeSearch}}
    \label{alg:0}
    \end{algorithm}
    
    In order to obtain the regret upper bound of $K$, we present Algorithm~\ref{alg:0} which iteratively improves the solution.
    Note that $E^{>}_{\tau}$ denotes the set of eligible arms with principal's utility greater than $\tau$, \ie $E^{>}_{\tau} := \{a\ |\ X_a > \tau \}$, and $E_{\tau}$ denotes the set of eligible arms with principal's utility at least $\tau$, \ie $E_{\tau} := \{a\ |\ X_a \geq \tau \}$.
    
    Due to the myopic and deterministic nature of the agent, it is obvious that Algorithm~\ref{alg:0} first finds the best arm in terms of the agent.
    Then, the eligible set is set to preclude that arm, and thus the agent will submit another arm which yields larger principal's utility.
    This process terminates at most after $K$ rounds since it excludes at least $1$ arm at each round.
    Thus, its regret is at most $O(K)$.
    Combining the above observations, we can essentially compare $K$ and $\log \log T$, and determine which algorithms to exploit, and it completes the proof.
\end{proof}

\section{Proof of Theorem~\ref{thm:alg1}}\label{sec:alg1}
\begin{algorithm}
    \SetAlgoLined
    $\tau \gets 0$, $D \gets \tgamma \log{\frac{\tgamma}{\ymin}}$, $f \gets 0$; \\
    \While{$f \neq 1$}
    {
        Announce $E^{>}_{\tau}$ for a single round. \\
        \leIf{solution $a$ is proposed by the agent}{
            $\tau' \gets X_a$
        }
        {$f \gets 1$}
        Announce $E^{>}_{\tau}$ for $D$ rounds. \\
        \lIf{$f \neq 1$}{$\tau \gets \tau'$}
    }
    Announce $E_{\tau}$ for remaining rounds.
\caption{\textsc{DelayedIterativeSearch}}\label{alg:1}
\end{algorithm}

\begin{proof}[Proof of Theorem~\ref{thm:alg1}]
    The proof relies on the following lemma.
    
    \begin{lemma}\label{lem:myopic_agent}
        In a $D$-delayed mechanism where $D = \tgamma \log{\frac{\tgamma}{\ymin}}$,
        when the principal announces $E^{>}_{\tau}$, 
        the agent submits a solution if and only if solution $a$ exists
        such that $X_a > \tau$.
    \end{lemma}
    \begin{proof}
        When the mechanism is $D$-delayed, then the agent plays his $\eps$ best response.
        Given $E^{>}_{\tau}$, if the agent proposes an eligible solution, then he obtains the utility of more than $\ymin$.
        On the other hand, if the agent does not propose any solutions, he gets the utility of $0$. 
        As $\eps = \ymin$, the $\eps$-best responding agent proposes a solution if he has at least one eligible solution.
    \end{proof}
    The above lemma indicates that by announcing $E^{>}_{\tau}$, the principal can check whether there exists solution $a$ such that $X_a > \tau$ or not.
    We consider algorithm~\ref{alg:1}.
    In this algorithm, when the principal observes solution $a$,
    she announces $E^{>}_{X_a}$ and due to Lemma~\ref{lem:myopic_agent}, checks whether there exists a better solution or not.
    By observing a better solution, the principal does the same thing and keeps observing better solutions.
    This process continues until the principal observes the best solutions.
    As a result, after at most $K$ iteration of the algorithm, the principal observes the best solution.
    Each iteration of the above algorithm, takes $D$ rounds, hence, total regret is bounded by $KD$ where $D = \tgamma \log{\frac{\tgamma}{\ymin}}$.
\end{proof}

\section{Proof of Theorem~\ref{thm:alg2_regret}}\label{sec:alg2}
\begin{algorithm}

    \SetAlgoLined
    $D \gets \tgamma \log{\frac{\tgamma}{\ymin}}$, $l \gets 0$, $r \gets 1$\\
    \While{$r-l > \frac{1}{T}$}
    {
        $\tau \gets \frac{l+r}{2}$; \\
        Announce $E^{>}_{\tau}$ for a single round. \\
        \leIf{solution $a$ is proposed by the agent}{
            $l \gets \tau$
        }{
            $r \gets \tau$
        }
        Announce $E_{l}$ for $D$ rounds. \\
    }
    Announce $E_{l}$ for remaining rounds. \\
\caption{\textsc{DelayedBinarySearch}}\label{alg:2}
\end{algorithm}

\begin{proof}[Proof of Theorem~\ref{thm:alg2_regret}]
    The proof is based on Algorithm~\ref{alg:2}.
    We claim that the optimal $X_{a^*}$ always lies in the interval $[l, r]$.
    When $l=0$ and $r=1$, it is obvious that the optimal solution lies in this interval.
    Then on each iteration, principal puts $\tau = \frac{l + r}{2}$ and announces $E^{>}_{\tau}$.
    According to Lemma~\ref{lem:myopic_agent}, if the agent has solution $a$ such that $X_a > \tau$, 
    the principal observes that solution and should find the optimal solution in the interval $[\tau, r]$.
    On the other hand, there do not exist any solutions in the interval $[\tau, r]$
    and the algorithm continues searching for the optimal solution in $[l, \tau]$.
    In this type of rounds, regret is at most $1$.
    
    The algorithm stops shrinking the interval until its length gets smaller than $\frac{1}{T}$. 
    After that the principal announces $E_{l}$ and as the optimal solution lies in the interval $[l, r]$, 
    regret is at most $r - l=\frac{1}{T}$.
    
    This algorithm needs $O(\log{T})$ iterations and on each iteration announces $O(D)$ eligible sets.
    On those rounds, regret is at most $r-l$. as $r-l$ is shrunk by a multiplicative factor of $2$, 
    total regret in these rounds will be at most
    \begin{align*}
        D + \frac{D}{2} + \frac{D}{4} + \frac{D}{8} + ... \leq 2D
    \end{align*}
    
    After $\log{T}$ iterations, regret is at most $\frac{1}{T}$. Hence,
    total regret is $O(\tgamma \log{\frac{\tgamma}{\ymin}} + \log{T})$.
\end{proof}

\section{Proof of Theorems~\ref{thm:linear_regret} and {\ref{thm:dense-lw}}}\label{sec:thm:linear_regret}

After establishing the proof for Theorem~\ref{thm:linear_regret}, it becomes straightforward to infer Theorem~\ref{thm:dense-lw}.
Thus, we start with Theorem~\ref{thm:linear_regret}.

\begin{proof}[Proof of Theorem~\ref{thm:linear_regret}]
    We consider two problem instances and claim that in at least one of them, the principal obtains a linear regret.
    We use $P_1$ and $P_2$ to refer to those instances.
    We assume that in $P_1$, there are two deterministic solutions $a_1$ and $a_2$ with utility of 
    $(X_{a_1}, Y_{a_1}) = (d, 1)$ and $(X_{a_2}, Y_{a_2}) = (2d, y)$.
    In $P_2$ the only solution is $a_1$.
    We assume that in both instances, the principal has already observed $a_1$.
    In problem $P_1$, the principal has to find $a_2$ as her regret is $d$ for each round of not receiving $a_2$.
    However, in problem $P_2$, $a_1$ is the optimal solution and the principal wastes her time by looking for better solutions.

    We show that under the setting where $y \leq e^{-T}$, the principal's regret in at least one of $P_1$ and $P_2$ will be $\Omega(Td)$.
    Let $C = \{c_1, c_2, ..., c_K\}$ be the rounds in instance $P_2$ such that the principal looks for a better solution, i.e., $a_2$. 
    In those rounds, she announces an eligible set such that $a_1$ is no longer eligible, thus, 
    in $P_2$ she receives no solutions and has the regret of $d$.
    Thus, her regret in $P_2$ is $\Omega(Kd)$ because, in $K$ rounds, she receives no solutions. This implies that
    in order to get sublinear regret, $K \in o(T)$.

    We note that in problem $P_1$, the principal needs to find $a_2$ in sublinear rounds.
    Indeed, if the principal observes $a_2$ after $l$ rounds, then her regret for previous rounds is at least $d$,
    hence, her total regret is $\Omega(ld)$, which implies that for sublinear regret $l$ should be sublinear, i.e., $l \in o(T)$.

    Now we consider the agent's response in $P_1$.
    We note that as long as the agent does not reveal $a_2$, the principal cannot distinguish between $P_1$ and $P_2$.
    So, her strategy is the same in both problem instances.
    In the case that the agent decides not to hide $a_2$, 
    denote by $j$ the first round that the agent's behavior is different in $P_1$ and $P_2$.
    In fact, in round $j$ of instance $P_1$, the agent proposes $a_2$.
    We note that according to the argument we mentioned above, $j \in O(l)$. This implies that $j \in o(T)$.
    After that round, as the principal has already observed $a_2$,
    she never announces an eligible set such that $a_1$ is accepted.
    Thus, agent's utility starting from round $j$ will be at most $\gamma^j y \sum_{t=j}^{T-1} \gamma^t \leq \gamma^j \frac{y}{1-\gamma}$.

    If the agent decides to hide $a_2$, then he does not propose it in round $j$ and the principal's strategy will be the same as her strategy for $P_2$.
    As the principal does not accept $a_1$ only for $K$ rounds, then in at least one of the rounds $j, j+1, ..., j+K$,
    principal accepts $a_1$ and the agent gets the utility of at least $\gamma^{j+K}$.
    We note that as both $j, K \in o(T)$, then all rounds $j, j+1, ..., j+K$ exist.

    Since the agent wants to maximize his utility, if he decides to propose $a_2$ in round $j$,
    then his maximum possible utility should be at least his minimum utility in the scenario he hides that solution.
    Formally,

    \begin{align*}
        \gamma^j \frac{y}{1-\gamma} \geq \gamma^{j+K}.
    \end{align*}

    This implies that $\frac{y}{1-\gamma} \geq \gamma^K$.
    By assuming the fact that $y \leq e^{-T}$, then we have
    \begin{align*}
        \frac{e^{-T}}{1-\gamma} \geq \gamma^K.
    \end{align*}
    By taking $\log$ of both sides, we get 
    \begin{align*}
        K\log{\gamma} \leq -T -\log{(1-\gamma)}.
    \end{align*}
    Finally, $K \geq -\frac{T}{\log{\gamma}} - \frac{\log{(1-\gamma)}}{\gamma}$.
    Hence, $K \in \Omega(T)$. 
    This implies that the principal needs to spend $\Omega(T)$ rounds to check whether a better solution exists or not, which is a contradiction
    with the previous argument that $K \in o(T)$.
\end{proof}

\begin{proof}[Proof of Theorem~\ref{thm:dense-lw}]
We note that the above instances $P_1$ and $P_2$ satisfy $d$-dense assumption.
This implies that we showed even under $d$-dense assumption when $\ymin \leq e^{-T}$, then there the principal suffers from $\Omega(Td)$ regret.
\end{proof}
\section{Proof of Theorem~\ref{thm:lips-lns-reg}}\label{sec:apd:3.8}
\begin{proof}
    We obtain an upper bound on the regret for a more generalized version of the algorithm~\ref{alg:5}. In this algorithm, we assume that there exists $\beta \geq 2$ such that 
    $L1 - \frac{\beta + 2}{2\beta} L_2 > 0$.

    \begin{algorithm}
    \SetAlgoLined
    \While{any solution has not been received}{
        Announce $E_0$.\\
    }
    Let $a_0$ be the proposed solution. \\
    $\alpha \gets L_1 - \frac{\beta+2}{2\beta}L_2$, $l \gets X_{a_0}$, $y \gets Y_{a_0}$; \\
    $r \gets \min\{1, l + \frac{y}{L_1}\}$; \\
    \While{$r-l > \beta d$}
    {
        $\tau \gets \frac{l+r}{2}$, $\eps \gets \alpha (r-l)$, $D = \tgamma \log{\frac{\tgamma}{\eps}}$; \\
        Announce $E^{>}_{\tau}$. \\
        \leIf{solution $a$ is proposed by the agent}{
            $l \gets X_a$, $y \gets Y_a$
        }
        {
            $r \gets \tau$
        }
        Announce $E_{l}$ for $D$ rounds. \\
        $r \gets \min\{r, l + \frac{y}{L_1}\}$; \\
    }
    Announce $E_l$ for remaining rounds.
    \caption{\textsc{DelayedProgeressiveSearch}}\label{alg:5}
    \end{algorithm}
        
    At first, the algorithm keeps accepting all solutions until the agent proposes any of them.
    it is obvious that there is no merit for the agent to not propose any solution at all. In fact, even after a single round of accepting all solutions, the agent proposes $a_0$ and we can initialize corresponding variables in our algorithm.
    
    We claim that the optimal $X_{a^*}$ always lies in the interval $[l, r]$.
    At first, as $l=0$ and $r=1$, this interval contains all solutions. 
    On each iteration, we assume that there exists solution $a_0$ such that $l = X_{a_0}$. 
    As long as $r-l > \beta d$, the algorithm shrinks the interval.
    let $p$ to be $\frac{r-l}{2}$, thus, $p \geq \frac{\beta}{2}d$.
    it picks $\tau = \frac{l+r}{2}$ and aims to check whether
    there exists solution $a$ such that $X_a \geq \tau$ or not.
    if it can guarantee there exists a solution in the interval $[\tau, r]$,
    the principal continues searching in that interval, 
    otherwise, she continues searching in the interval $[l, \tau]$.
    
    According to $d$-dense assumption, if there exists a solution $a$ such that $X_a \geq \tau$,
    then there also exists a solution $a'$ such that $\tau \leq X_{a'} \leq \tau + d$. 
    We note that $\tau + d \leq r$ since $\beta \geq 2$.
    As $X_{a'} \leq \tau + d$,
    according to Lipschitz continuity $d_Y(a', a_0) \leq L_2 \cdot (X_{a'} - X_{a_0})$, i.e., 
    \begin{align*}
        Y_{a_0} - Y_{a'} 
        &\leq L_2 \cdot (X_{a'} - X_{a_0})\\
        &\leq L_2 \cdot (\tau + d - l)\\ 
        &\leq L_2 \cdot (p + d)\\
        &\leq L_2 \cdot \left(\frac{2+\beta}{\beta} p\right)
    \end{align*}
    This implies that 
    \begin{align}
        Y_{a'} \geq Y_{a_0} - L_2 \cdot (\frac{2+\beta}{\beta} p).
        \label{eq:1684260763}
    \end{align}

    Furthermore, we can always assume that $r-l$ is bounded.
    This is due to the fact that if there is a solution $a''$ with $X_{a''} \geq r$, 
    then according to Lipschitz continuity, $d_Y(a'', a_0) \geq L_1 \cdot (X_{a''} - X_{a_0})$, i.e., 
    \begin{align*}
        |Y_{a_0} - Y_{a''}| 
        &\geq 
        L_1 \cdot (X_{a''} - X_{a_0}) 
        \\
        &\geq 
        L_1 \cdot (r - l)
        \\
        &\geq 
        2p \cdot L_1.
    \end{align*}
    This implies that either $Y_{a''} \leq Y_{a_0} - 2p L_1$ or $Y_{a''} \geq Y_{a_0} + 2p L_1$.

    We first show that $Y_{a''} \geq Y_a + 2p L_1$ is not possible.
    As we will later show, we always put $\epsilon = \alpha \beta d$
    where $\alpha = L_1 - \frac{\beta + 2}{\beta}L_2$.
    Thus, the agent always plays $\alpha \beta d$ best response.
    In this case, $X_{a''} > X_{a_0}$ and $Y_{a''} - Y_{a_0} \geq 2p L_1 \geq L_1 \beta d > \alpha \beta d$.
    This implies that $\alpha \beta d$ best-responding agent
    should not propose $a_0$ as $a''$ gives him at least $\alpha \beta d$ more utility and
    whenever solution $a_0$ is eligible, so does $a''$.

    Therefore, we focus on case where $Y_{a''} \leq Y_{a_0} - 2p L_1$. 
    if $Y_{a_0} - 2p L_1 < 0$, then there cannot be any solution with $X_{a''} \geq r$,
    thus, we can shorten the interval, \ie we can decrease $r$.
    In fact,
    \begin{align}
        Y_{a_0} - 2p L_1 \geq 0,
        \label{eq:1684260948}
    \end{align}
    indicating that $r$ should not exceed $l + \frac{Y_{a_0}}{L_1}$. Consequently, in each iteration, we verify if $r$ is smaller than $l + \frac{Y_{a_0}}{L_1}$. If it surpasses this value, we simply update it to $l + \frac{Y_{a_0}}{L_1}$.

    Now, we use this lower bound on $Y_{a_0}$ to find a lower bound on $Y_{a'}$.
    We define $\alpha > 0$ to be $\alpha := L_1 - \frac{\beta + 2}{\beta}L_2$.
    
    \begin{align*}
        Y_{a'} &\geq Y_{a_0} - \frac{\beta + 2}{\beta}p \cdot L_2 \tag{\ref{eq:1684260763}} \\
        &\geq 2p \cdot L_1 - \frac{\beta + 2}{\beta}p \cdot L_2 \tag{\ref{eq:1684260948}}\\
        &\geq 2p \left(L_1 - \frac{\beta + 2}{2\beta}L_2\right) \\
        &\geq \alpha \beta d. \tag{definition of $\alpha$}
    \end{align*}

    The above lower bound shows that if there exists a solution in the interval $[\tau, r]$,
    then there exists solution $a'$ such that the agent's utility for that is at least $\alpha \beta d$.
    We put $\eps$ to be $\alpha \beta d$ and $D$ to be $\tgamma \log{\frac{\tgamma}{\eps}}$, 
    Then the $D$-delayed mechanism makes the agent play his $\eps$ best response. 
    Hence, by announcing $E^{>}_{\tau}$, the agent will propose a solution if there exists a solution in interval $[\tau, r]$.

    After observing his proposal, the algorithm updates the interval correspondingly.
    If $a$ is proposed by the agent ($X_a \geq \tau$), then it keeps searching for better solutions in the interval $[X_a, r]$.
    We note that for the next iteration, we maintain the property $l = X_{a}$.
    On the other hand, if the agent does not propose any solutions, the principal has to search for better solutions in the interval $[l, \tau]$. Similarly, note that $l=X_{a_0}$ in the next iteration. in this type of rounds, regret is at most $1$.
    
    We keep shrinking this interval until its length gets smaller than $\beta d$ ($\beta \geq 2$).
    At this point, we keep announcing $E_{l}$ and solution $a_0$ is eligible, so the principal gets the utility of at least $l$
    and the regret is at most $r-l \leq \beta d$.

    As in each iteration, the length of the interval becomes at most half of the previous interval, 
    number of iterations is $O(\log{\frac{1}{\beta d}})$.
    On each iteration, we run $O(D)$ rounds where
    \begin{align*}
        D = \tgamma \log{\frac{\tgamma}{\alpha \beta d}}.
    \end{align*}
    The regret of the algorithm on those rounds is at most $r-l$.
    As $r-l$ is shrunk by at least a multiplicative factor $2$, then the total regret in these 
    rounds as at most
    \begin{align*}
        D + \frac{D}{2} + \frac{D}{4} + ... \leq 2D.
    \end{align*}
    Also, after stopping exploration, the regret is at most $\beta d$. So the total regret is 
    $O(\tgamma \log{\frac{\tgamma}{\alpha \beta d}} +  \log{\frac{1}{\beta d}} + \beta d T)$.
    By plugging $\beta = 4$, we get the bound provided in the Theorem~\ref{thm:lips-lns-reg}.
\end{proof}

\section{Proof of Theorem~\ref{thm:stoc_myop}}\label{sec:apd:I}
\begin{proof}
	Our proof is essentially based on the following series of lemmas.

	\begin{lemma}\label{lm:klein1}
		There exist constants $C_1,C_2$ such that $C_1(\tau^* - \tau)^2 < f(\tau^*) - f(\tau) < C_2(\tau^* - \tau)^2$.
	\end{lemma}
	\begin{proof}
		Since $f''(\tau^*)$ exists and has a strictly negative value, there exists a neighbor $N_\eps(\tau^*) = (\tau^* - \eps, \tau^* + \eps)$ such that for some constants $A_1,A_2$ we have $A_1(\tau^* - \tau)^2 < f(
	\tau^*) - f(\tau) < A_2(\tau^* - \tau)^2$ for $\tau \in N_\eps(\tau^*)$.
		Now consider the complement of $N_\eps(\tau)$ defined by $X = \{\tau \in [0,1]: |\tau^* - \tau| \ge \eps\}$.
		Due to the compactness of $X$ and since $f(\tau^*) - f(\tau)$ is strictly positive for all $x \in X$, there are constants $B_1,B_2$ such that $B_1(\tau^* - \tau)^2 < f(\tau^*) - f(\tau) < B_2(\tau^* - \tau)^2$ for $x \in X$.
		Then, $C_1 =\min(A_1,B_1)$ and $C_2 = \max(A_2, B_2)$ finishes the proof.
	\end{proof}

	\begin{lemma}\label{lm:discretization}
		The discretization error $f(\tau^*) - \max_i f(i/K) \le C_2/K^2$.
	\end{lemma}
	\begin{proof}
		Note that there exists an element in the set of arms $\{1/K, \ldots, K/K\}$ satisfy $i/K - x^* \le 1/K$, denote this arm by $j$.
		Then we have
		\begin{align*}
			f(\tau^*) - \max_i f(i/K) &\le f(\tau^*) - f(j/K) \\ &\le
            C_2(\tau^* - j/K) \\ &\le C_2/K^2,
		\end{align*}
		where the second inequality follows from Lemma~\ref{lm:klein1}, and the last holds due to our construction of $j$.
	\end{proof}

	\begin{lemma}\label{lm:inversegapbound}
		Let $\Delta_i = \max_{1\le j \le K} f(j/K) - f(i/K) $.
		Then, The sum of inverse gaps satisfies $\sum_{\Delta_i > 0}1/\Delta_i \le 7K^2/C_1$.
	\end{lemma}
	\begin{proof}
		Due to Lemma~\ref{lm:klein1}, we have $\Delta_i \ge C_1(\tau^* - i/K)^2$ for $i \in [K]$.
		Define $0 =\tilde{\Delta}_0 \le \tilde{\Delta}_1 \le \ldots \le \tilde{\Delta}_{K-1}$ be an ordered sequence of the set $A = \{\Delta_i\}_{i \in [K]}$.
		Since at most $j$ elements of the set $A$ satisfy $i/K - x^* \le j/2K$, we have $\tilde{\Delta}_j \ge C_1(j/2K)^2$.
		Hence we obtain
		\begin{align*}
			\sum_{i: \Delta_i \neq 0}1/\Delta_i
			<
			\sum_{i=1}^K C_1^{-1}(i/2K)^{-2} 
			&< 
			\frac{4K^2}{C_1}\sum_{i=1}^K i^2\\
			&<
			\frac{4K^2}{C_1}\cdot \frac{\pi^2}{6}\\
			&\le
			\frac{7K^2}{C_1},
		\end{align*}
            and it completes the proof of the lemma.
	\end{proof}
    We can now use the following regret bounds for the UCB1 algorithm to complete the proof.
    \begin{lemma}\label{lm:banditth}
        \cite{bandit2002}
        In a multi-armed bandit instance with $K>1$ arms, the UCB1 algorithm achieves an expected regret of at most 
        \begin{align*}
            O([\sum\limits_{\Delta_i>0}\frac{\log{T}}{\Delta_i}]+\sum\limits_{i=1}^{K}\Delta_i)
        \end{align*}
        after $T$ rounds, where $\Delta_i=\max_{1\le j \le K}f(j/K)-f(i/K)$ is the difference between the expected utility of arm $i$ and the optimal arm.
    \end{lemma}
    \begin{lemma}\label{lm:discretereg}
    Let $DiscreteReg$ be the regret of the UCB1 algorithm compared to the utility of always using $\tau = [arg\max_{1\le i \le K}{f(i/K)]/K}$ as the threshold. Then $DiscreteReg = O(\frac{K^2}{C_1}\log{T}+K)$ 
    \end{lemma}
    \begin{proof}
    Using Lemma \ref{lm:banditth}, we have
    \begin{align*}
        DiscreteReg &= O([\sum\limits_{\Delta_i>0}\frac{\log{T}}{\Delta_i}]+\sum\limits_{i=1}^{K}\Delta_i) \tag{Lemma\ref{lm:banditth}} \\
        &=O(\frac{K^2}{C_1}\log{T}+\sum\limits_{i=1}^{K}\Delta_i) \tag{Lemma \ref{lm:inversegapbound}}\\
        &=O(\frac{K^2}{C_1}\log{T}+K) \tag{$\Delta_i \le 1$}.
    \end{align*}
    \end{proof}
    The regret of UCB1 compared to the optimal benchmark, which achieves a utility of $T\max{\tau}f(\tau)$, can be calculated by summing up $DiscreteReg$ and the discretization error. Therefore, the total regret of the algorithm is
    \begin{align}
        REG &= DiscreteReg + DiscretizationError\\
            &= O(\frac{K^2}{C_1}\log{T}+K + \frac{T}{K^2}), \label{ineq:05170137}
    \end{align}
    where we use Lemma \ref{lm:discretization} and Lemma \ref{lm:discretereg}.
    Plugging in $K={(\frac{T}{\log{T}})}^{\frac{1}{4}}$, we get the desired bound on the regret
    \begin{align*}
        REG &= O(\frac{K^2}{C_1}\log{T}+K + \frac{T}{K^2})\\
        &= O(\frac{\sqrt{T\log{T}}}{C_1} + {(\frac{T}{\log{T}})}^{\frac{1}{4}} + \sqrt{T\log{T}})\\
        &=O(\sqrt{T\log{T}}).
    \end{align*}
\end{proof}

\section{Proof of Theorem~\ref{thm:stoc_stra}}\label{sec:apd:stoc}

\begin{algorithm}
    \SetAlgoLined
    $K$ arms, delay $D$, perturbation $\delta$, $\set{A} \gets [K]$\\
    \While{$t < D$}
    {
    	\For{$i \in \set{A}$}
    	{
    		Pull arm $i$ and observe feedback (principal's utility) $r_t$;\\
    		$t \gets t+1$;\\
    	}    	
        \For{$i \in \set{A}$}
        	{
        		
        		\leIf{$t - D \ge 1$}
        		{
        			\\
        			$n \gets \max(\sum_{\tau =1}^{t-D}\Ind{i_\tau = 1},1)$;\\
        			$\hat{\mu}_i \gets \frac{1}{n}\sum_{\tau=1}^{t-D}\Ind{i_\tau = i}r_\tau$;\\
        			$\text{LCB}_i \gets \hat{\mu}_i - \sqrt{2 \log(T)/n} - \delta$; \\
        			$\text{UCB}_i \gets \hat{\mu}_i + \sqrt{2 \log(T)/n} + \delta$; \\
        		}
        		{continue}
        	}
        	$\set{A} \gets \{i\in \set{A}: \text{UCB}_i \ge \text{LCB}_j, \forall j \in \set{A}\}$.
    }

\caption{\textsc{SuccessiveEliminationDelayed} [\cite{haghtalab2022learning}]}\label{alg:4}
\end{algorithm}
Before we proceed, we present the formal definition of the perturbed bandit problem.
\begin{definition}[Bandits with perturbed rewards]
    Let $\cI$ be an instance of the multi-armed bandit problem.
    We say that the ex-post sequence of rewards $r_1,r_2,\ldots,r_T$ are $\delta$-perturbed from the $\cI$ if for any arm sequence $a_1,\ldots, a_T$ that happens with positive probability in $\cI$, there exist independent random intervals $\{[l_t, r_t]\}_{t \in [T]}$ such that $r_t \in [l_t,u_t]$ almost surely and $\mu_{a_t} - \delta \le \Ex{l_t} \le \Ex{u_t} \le \mu_{a_t} + \delta$.
\end{definition}

\begin{proof}
        We mainly analyze the regret upper bound of Algorithm~\ref{alg:4}.
        We first present the following lemma, which characterizes the regret upper bound of Algorithm~\ref{alg:4} for a bandit problem with delayed feedback and perturbed rewards.
        \begin{lemma}[\cite{haghtalab2022learning}]\label{lm:perturbed}
	   For a multi-armed bandit problem with $K$ arms, $\delta$-perturbed rewards, and $D$-delayed feedback, Algorithm~\ref{alg:4} has regret upper bound of $O(\sum_{\Delta_i >0} \frac{\log T}{\Delta_i} +\delta T + D \log K)$, where $\Delta_i$ denotes the difference of the mean rewards between the optimal arm and arm $i$.
        \end{lemma}
        
	To apply Lemma~\ref{lm:perturbed}, we essentially need to construct the proper random interval to guarantee some bounded perturbation of our reduced bandit problem.
	We write $E_\tau$ to denote the eligible set $\{a: X_a \ge \tau\}$.
        Our construction of perturbation proceeds as follows.
	\begin{lemma}
		Fix an ex-post realization of all the arms.
		Let $x_{\tau} = X_{\br(E_\tau)}\Ind{X_{\br(E_\tau)} \ge \tau}$ and define $l = (x_\tau - L_1^{-1}\eps)$ and $u = (x_\tau + L_1^{-1}\eps)$.
		If $a \in \br_\eps(E_\tau)$, then $l \le X_{a}\Ind{X_a \ge \tau} \le u$.
		In addition, we have $\Ex{x_\tau} - L_1^{-1}\eps \le \Ex{l} \le \Ex{u} \le \Ex{x_\tau} + L_1^{-1}\eps$.
	\end{lemma}
	\begin{proof}
		The latter condition directly follows from the definition of $l$ and $u$. 
		To check the former inequality, we use our assumptions on $y_{min}$ and Lipschitz continuity.
  
        Let $a \in \br_\eps(E_\tau)$ be an $\eps$-best response arm. If arm $a$ is not in the eligible set, then the agent must not have any eligible arms. Otherwise, proposing that arm would achieve a utility of at least $y_{min}>\eps$, which contradicts $a$ being an $\eps$-best response. Therefore in this case, $X_a\Ind{X_a \ge \tau}$ and $x_\tau$ are both equal to $0$ and so the first condition stands.

        Now, we can assume that $a$ is an eligible arm. We have
		\begin{align*}
			|X_a - X_{\br(E_\tau)}| &\le L_1^{-1}|Y_a - Y_{\br(E_\tau)}|\\
            &\le \eps L_1^{-1},
		\end{align*}
            where the first inequality follows from Lipschitz continuity, and the second holds by $\eps$-best-responseness of the action $a$.
		Thus, we obtain
		\begin{align*}
		 	X_{\br(E_\tau)} - \eps/L_1
		 	\le 
		 	X_a
		 	\le
		 	X_{\br(E_\tau)} + \eps/L_1.
		\end{align*}
		If arm $a$ is eligible, then it implies that the best response is also eligible.
        Therefore, we have
            \begin{align*}
			X_a\Ind{X_a \ge \tau} 
			&= 
		  X_a\\  
			&\le
			X_{\br(E_\tau)} + \eps/L_1 \\
           &= (X_{\br(E_\tau)}\Ind{{\br(E_\tau)} \ge \tau} + \eps/L_1)\\
           &= u,
		\end{align*}
            where the first equation follows from $a$'s eligibility and second equation follows since $\br(E_\tau)$ is eligible.
		Further we similarly obtain $X_a \ge (X_{\br(E_\tau)}\Ind{{\br(E_\tau)} \ge \tau} - \eps/L_1) = l$.
  
		Now observe that
		\begin{align*}
			\Ex{l}
			=
			\Ex{(x_\tau - L_1^{-1}\eps)}
			&=
			\Ex{x_\tau}	 - L_1^{-1}\eps.
		\end{align*}
		Similarly,
		\begin{align*}
			\Ex{u}
            =
			\Ex{(x_\tau + L_1^{-1}\eps)}
			&=
			\Ex{x_\tau}	 + L_1^{-1}\eps.
		\end{align*}
		and it completes the proof.
	\end{proof}
	
	Together with Lemma~\ref{lm:perturbed}, our bandit problem is $L_1^{-1}\eps$-perturbed from the stochastic setting with a myopic agent, assuming that the agent is $\eps$-best responding.
	Hence, similar to \eqref{ineq:05170137} in the proof of Theorem~\ref{thm:stoc_myop}, running Algorithm~\ref{alg:4} with $\delta = L_1^{-1}\eps$ yields
	\begin{align*}
		\reg_T
		=
		O(T\log T \cdot \frac{7K^2}{C_1} +L_1^{-1}\eps T + D \log K + \frac{TC_2}{K^2}).
	\end{align*}
	Plugging in $\eps = \min(L_1/T,y_{min})$, $D = T_\gamma \log (T_\gamma / \eps)$, and $K = (T/\log T)^{1/4}$, we obtain
	\begin{align*}
		\reg_T
		=
		O\Big(&(C_1+C_2^{-1})(T\log T)^{1/2} \\ &+
        T_\gamma \log(T_\gamma \max(T/L_1,1/y_{min})\log{T}\Big),
	\end{align*}
	and it completes the proof.
\end{proof}

\section{Necessity of Discount Factor}\label{sec:discount}
We note that with modest assumptions about the difference in the agent’s utility for different solutions, we can show that it is not possible for the principal to achieve a sublinear regret.  To show the inevitability of linear regret, we can consider two simple cases with a strategic agent with deterministic utility. In the first case, the agent only has a single solution with utility $1/2$  for the principal and utility $1$ for the agent. To achieve sublinear regret compared to the optimal, the principal needs to accept this solution in all but a sublinear number of rounds. For a large enough $T$, we can assume that this solution needs to be accepted in at least $2T/3$ rounds.

In the other case, the agent again has the same solution with utilities $1$ and $1/2$, in addition to a solution with a utility of $1$ for the principal and utility $\epsilon$ for the agent. For the principal to achieve a sublinear regret, she needs to get the solution with utility $1$ in all but a sublinear number of rounds.
So, we can assume that for large enough $T$, any mechanism achieving sublinear regret gets the agent to submit this solution in at least $2T/3$ rounds. This results in a utility of at most $2T/3 \cdot \eps + T/3$ for the agent. However, if the agent pretends to only have the first solution and acts the same as the agent in the other case, he can achieve a utility of at least $2T/3$, since the principal wants to have sublinear regret in the previous case. For $\eps < 1/2$, this results in a better utility for the agent. Therefore, it is impossible for the principal to achieve a sublinear regret for both of these cases for large $T$.

\end{document}